\documentclass{llncs}
\usepackage{paralist}
\usepackage{amsmath,amssymb,graphicx,epsfig}
\usepackage{algorithm}
\usepackage{algorithmic}
\usepackage{subfigure}
\usepackage{luca}
\usepackage{macro}
\usepackage{macros,comment}
\usepackage{times}

\begin{document}

\def\dd#1{{\mathsf{#1}}}
\def\ddop#1{{\mbox{\sc{#1}}}}

\title{Magnifying Lens Abstraction for Stochastic Games with
Discounted and Long-run Average Objectives}

\author{Krishnendu Chatterjee\inst{1} \and Luca de Alfaro\inst{2} \and Pritam Roy\inst{3}}

\institute{ IST Austria 
\and UC Santa Cruz, USA
\and  UC Los Angeles, USA
}
\date{}
\maketitle
\thispagestyle{empty}
\begin{abstract}
Turn-based stochastic games and its important subclass Markov decision 
processes (MDPs) provide models for systems with both probabilistic and 
nondeterministic behaviors.
We consider turn-based stochastic games with two classical quantitative 
objectives: discounted-sum and long-run average objectives.
The game models and the quantitative objectives are widely used in 
probabilistic verification, planning, optimal inventory control, network 
protocol and performance analysis.
Games and MDPs that model realistic systems often have very large 
state spaces, and probabilistic abstraction techniques are necessary to handle 
the state-space explosion.
The commonly used full-abstraction techniques do not yield space-savings for 
systems that have many states with similar value, but does not necessarily 
have similar transition structure. 
A semi-abstraction technique, namely  Magnifying-lens abstractions (MLA), that 
clusters states based on value only, disregarding differences in their 
transition relation was proposed for qualitative objectives (reachability and
safety objectives)~\cite{deAlR07}.
In this paper we extend the MLA technique to solve stochastic games with 
discounted-sum and long-run average objectives.
We present the MLA technique based abstraction-refinement algorithm for 
stochastic games and MDPs with discounted-sum objectives.
For long-run average objectives, our solution works for all MDPs and a 
sub-class of stochastic games where every state has the same value.
\end{abstract}


\section{Introduction}
\noindent
A \emph{turn-based stochastic game} is played on a finite graph with three
types of states:
in player-1 states, the first player chooses a successor state from a given
set of outgoing edges;
in player-2 states, the second player chooses a successor state from a given
set of outgoing edges;
and in probabilistic states, the successor state is chosen according to a given
probability distribution.
The game results in an infinite path through the graph.
An important subclass of turn-based stochastic games is \emph{Markov decision
processes (MDPs)}: in MDPs the set of player-2 states is empty.
Turn-based stochastic games and MDPs provide models for the study of dynamic 
systems that exhibit both probabilistic and nondeterministic behavior.

Turn-based stochastic games with \emph{qualitative objectives} such as 
reachability, safety, and more general $\omega$-regular objectives has 
been widely studied in literature \cite{CY95,luca-thesis,CdAH-icalp05,Condon92} 
in the context of verification of probabilistic systems. 
Many other application scenarios such as planning, inventory control,
performance analysis require the study of turn-based stochastic games with 
\emph{quantitative objectives}~\cite{Derman,Bertsekas95,SegalaT,luca-thesis,KNP04a}.  
The two classical quantitative objectives studied in literature are as follows:
\emph{discounted-sum} (in short, discounted) and \emph{long-run average} 
objectives~\cite{FilarVrieze97,Bertsekas95}. 
In both these objectives a real-valued reward is assigned to every state.
For an infinite path (infinite sequence of states in the game graph), 
the discounted objective assigns a payoff that is the 
discounted sum of the rewards that appear in the infinite path, and the
long-run average objective assigns the long-run average of the rewards
that appear in the path.
Turn-based stochastic games and MDPs with discounted and long-run average 
objectives provide an important and powerful framework for studying a wide 
range of applications~\cite{FilarVrieze97,Bertsekas95}.

Turn-based stochastic games and MDPs that model realistic systems typically
have very large state spaces. 
Therefore the main algorithmic challenge in analysing such models consist of 
developing algorithms that work efficiently on large state spaces. 
In the non-probabilistic setting, abstraction techniques have been
successful in coping with large state-spaces~\cite{CounterexCAV00}.
By ignoring details not relevant to the property under study, abstraction makes 
it possible to answer questions about a system through the analysis of a smaller, 
more concise abstract model. 
The abstraction-refinement techniques for non-probabilistic setting do not 
always have a straight-forward extension to the probabilistic models.
The commonly used full-abstraction techniques do not yield space-savings for 
systems that have many states with similar value, but not necessarily 
have similar transition structure. 
A semi-abstraction technique, namely  
\emph{Magnifying-lens abstractions (MLA)}, 
was proposed for a subclass of qualitative objectives (namely, 
reachability and safety objectives)~\cite{deAlR07}.
MLA is a semi-abstract technique that can cluster states based on value 
only and can disregard the differences in their transition relation.
MLA is particularly well-suited to problems where there is a notion of 
{\em locality\/} in the state space, so that it is useful to cluster 
states based on values, even though their transition relations 
may not be similar. 
Many inventory, planning and control problems satisfy the locality 
property and would benefit from the MLA technique.
In the setting of inventory, planning and control problems 
quantitative objectives are more appropriate than qualitative 
objectives. 
This provides a strong and practical motivation for extending the work 
of~\cite{deAlR07} to provide MLA technique based solution for 
turn-based stochastic games and MDPs with quantitative objectives.

In this paper we extend the MLA technique to solve stochastic games with 
quantitative objectives. 
The MLA technique of~\cite{deAlR07} works for MDPs and the special 
class of qualitative objectives, namely reachability and safety 
objectives (the model is quantitative with probabilities 
but the objectives are qualitative).
We present the MLA technique based abstraction-refinement algorithm for 
both stochastic games and MDPs with discounted objectives.
For long-run average objectives, our solution works for all MDPs and a 
sub-class of stochastic games where every state has the same value.
We note that for long-run average objectives in stochastic games, 
the same assumption (of all states having the same value) is required 
for the \emph{relative value} iteration algorithm to 
work~\cite{Bertsekas95,FilarVrieze97}\footnote{Thus the assumption is necessary even for classical value iteration algorithms and even without abstraction, 
and hence cannot be avoided in our setting with 
abstraction.}.
Hence our result present generalizations of the results of~\cite{deAlR07}
from the sub-class of reachability and safety objectives (which are 
Boolean) to the general class of discounted and long-run average 
objectives (which are quantitative).
An abstraction-refinement based technique was proposed in~\cite{CHJM05}
for turn-based stochastic games with quantitative objectives, but the 
technique of~\cite{CHJM05} does not provide either a useful way to abstract 
probabilities, or the space-saving benefit of the MLA
based technique.
Thus our algorithms provide space-efficient and practical algorithmic 
solutions for a wide class of problems of interest.
To demonstrate the applicability of our algorithms we present a symbolic
implementation of our algoritms for MDPs with discounted objectives.
In Section~\ref{sec-ex} we present many examples to illustrate 
cases where MLA based solution has a clear advantage over the full abstraction
techniques, and our experimental results show that the 
MLA based technique gives a significant space saving.


\section{Preliminaries}\label{sec-defs}
\noindent
For a finite set $S$, a {\em probability distribution\/} on $S$ is a 
function $p : S \to [0,1]$ such that $\sum_{s \in S} p(s) = 1$;
we denote the set of probability distributions on $S$ by $\distr(S)$.
A \emph{valuation} over a set $S$ is a function $v : S \to \reals$
associating a real number $v(s)$ with every $s \in S$. 
For $x \in \reals$, we denote by $\mathbf{x}$ the valuation with
constant value $x$; for $T \subs S$, we indicate by $[T]$ the
valuation having value~1 in $T$ and~0 elsewhere.  
For two valuations $v,u$ on $S$, we define 
$||v-u|| = \sup_{s \in S} |v(s) - u(s)|$.

A {\em partition\/} of a set $S$ is a set $R \subs 2^S$, such that
$\bigcup_{x \in R} \set{s | s \in x} = S$ and $x \inters x' = \emptyset$ for all $x \neq x' \in R$. 
For $s \in S$ and a partition $R$ of $S$, we denote by $[s]_R$ the
element $x \in R$ with $s \in x$. 
We say that a partition $R$ is {\em finer\/} than a partition $R'$ if
for any $x \in R$ there exists $x'\in R'$ such that $x \subseteq x'$. 

We consider the class of turn-based probabilistic games and its important 
subclass of Markov decision processes (MDPs).

\smallskip\noindent{\bf Game graphs.} 
A \emph{turn-based probabilistic game graph} 
(\emph{$2\half$-player game graph}) $G =((S, E), (\SA,\SB,\SR),\trans)$ 
consists of a directed graph $(S,E)$, a partition $(\SA$, $\SB$,$\SR)$ of the finite set 
$S$ of states, and a probabilistic transition function $\trans$: $\SR \rightarrow \distr(S)$, 
where $\distr(S)$ denotes the set of probability distributions over the state space~$S$. 
The states in $\SA$ are the {\em player-$\PA$\/} states, where player~$\PA$
decides the successor state; the states in $\SB$ are the {\em 
player-$\PB$\/} states, where player~$\PB$ decides the successor state; 
and the states in $\SR$ are the {\em probabilistic\/} states, where
the successor state is chosen according to the probabilistic transition
function~$\trans$. 
We assume that for $s \in \SR$ and $t \in S$, we have $(s,t) \in E$ 
iff $\trans(s)(t) > 0$, and we often write $\trans(s,t)$ for $\trans(s)(t)$. 
For technical convenience we assume that every state in the graph 
$(S,E)$ has at least one outgoing edge.
For a state $s\in S$, we write $E(s)$ to denote the set 
$\set{t \in S \mid (s,t) \in E}$ of possible successors.
For $s \in S$ and a partition $R$ of $S$, a region $r_2 \in R$ is called successor 
to a region $r_1 \in R$ if at least one concrete state in $r_1$ has non-zero 
probability to reach concrete state(s) in $r_2$.
%
The \emph{Markov decision processes} (\emph{$1\half$-player game graphs}) 
are the special case of the $2\half$-player game graphs with 
$\SA = \emptyset$ or $\SB = \emptyset$. 
We refer to the MDPs with $\SB=\emptyset$ as \emph{player-$\PA$} MDPs,
and to the MDPs with $\SA=\emptyset$ as \emph{player-$\PB$} MDPs.

\smallskip\noindent{\bf Plays and strategies.}
An infinite path, or a \emph{play}, of the game graph $G$ is an 
infinite 
sequence $\pat=\seq{s_0, s_1, s_2, \ldots}$ of states such that 
$(s_k,s_{k+1}) \in E$ for all $k \in \Nats$. 
We write $\Paths$ for the set of all plays, and for a state $s \in S$, 
we write $\Paths_s\subseteq\Paths$ 
for the set of plays that start from the state~$s$.
A \emph{strategy} for  player~$\PA$ is a function 
$\straa$: $S^*\cdot \SA \to \distr(S)$ that assigns a probability 
distribution to all finite sequences $\vec{w} \in S^*\cdot \SA$ of states 
ending in a player-1 state 
(the sequence represents a prefix of a play).
Player~$\PA$ follows the strategy~$\straa$ if in each player-1 
move, given that the current history of the game is
$\vec{w} \in S^* \cdot \SA$, she chooses the 
next state according to the probability distribution $\straa(\vec{w})$.
A strategy must prescribe only available moves, i.e., 
for all $\vec{w} \in S^*$,
$s \in \SA$, and $t \in S$, if $\straa(\vec{w} \cdot s)(t) > 0$, then 
$(s, t) \in E$.
The strategies for player~2 are defined analogously.
We denote by $\Straa$ and $\Strab$ the set of all strategies for player~$\PA$
and player~$\PB$, respectively.

Once a starting state  $s \in S$ and strategies $\straa \in \Straa$
and $\strab \in \Strab$ for the two players are fixed, the outcome
of the game is a random walk $\pat_s^{\straa, \strab}$ for which the
probabilities of events are uniquely defined, where an \emph{event}  
$\Aa \subseteq \Paths$ is a measurable set of plays. 
For a state $s \in S$ and an event $\Aa\subseteq\Paths$, we write
$\Prb_s^{\straa, \strab}(\Aa)$ for the probability that a play belongs 
to $\Aa$ if the game starts from the state $s$ and the players follow
the strategies $\straa$ and~$\strab$, respectively.
For a measurable function $f:\Paths \to \reals$ we denote by 
$\Exp_s^{\straa,\strab}[f]$ the \emph{expectation} of the function
$f$ under the probability measure $\Prb_s^{\straa,\strab}(\cdot)$.

Strategies that do not use randomization are called pure.
A player-1 strategy~$\straa$ is \emph{pure} if for all $\vec{w} \in S^*$
and $s \in \SA$, there is a state~$t \in S$ such that  
$\straa(\vec{w}\cdot s)(t) = 1$. 
A \emph{memoryless} player-1 strategy does not depend on the history of 
the play but only on the current state; i.e., for all $\vec{w},\vec{w'} \in 
S^*$ and for all $s \in S_1$ we have 
$\straa(\vec{w} \cdot s) =\straa(\vec{w}'\cdot s)$.
A memoryless strategy can be represented as a function 
$\straa$: $\SA \to \distr(S)$.
A \emph{pure memoryless strategy} is a strategy that is both pure and 
memoryless.
A pure memoryless strategy for player~1 can be represented as
a function $\straa$: $\SA \to S$.
We denote by 
$\Straa^{\PM}$ the set of pure memoryless strategies for player~1.
The pure memoryless player-2 strategies $\Strab^{\PM}$ are
defined analogously.

\smallskip\noindent{\bf Quantitative objectives.} A \emph{quantitative} objective
is specified as a measurable function $f:\Paths \to \reals$.
We consider \emph{zero-sum} games, i.e., games that are 
strictly competitive.
In zero-sum games the objectives of the players are functions $f$ and 
$-f$, respectively. 
We consider two classical quantitative objectives specified as 
discounted sum objective and long-run average (mean-payoff) objectives. 
The definitions of are as follows. 
\begin{itemize}

\item 
\emph{Discounted objectives.} 
Let $r:S\to \rgz$ be a real-valued reward function that assigns to 
every state $s$ the reward $r(s)$, and let $0<\beta<1$ be a discount factor.
The \emph{discounted} objective $\Disc$ assigns to every play the $\beta$-discounted 
sum of the rewards that appears in the play.
Formally, for a play $\pat=\seq{s_0,s_1,s_2,s_3,\ldots}$ we have
$\Disc(\beta,r)(\pat)= \sum_{i=0}^\infty \beta^i\cdot r(s_i)$.

\item \emph{Long-run average objectives.} 
Let $r:S\to \rgz$ be a real-valued reward function that assigns to 
every state $s$ the reward $r(s)$.
The \emph{long-run} average objective $\LimAvg$ assigns to every play the long-run 
average of the rewards that appear in the play.
Formally, for a play $\pat=\seq{s_1,s_2,s_3,\ldots}$ we have
$\LimAvg(r)(\pat)=\liminf_{T \to \infty} \frac{1}{T} \cdot \sum_{i=0}^{T-1} r(s_i)$.
\end{itemize}

\smallskip\noindent{\bf Values and optimal strategies.}
Given a game graph $G$, and quantitative objectives specified as 
measurable functions $f$ and $-f$ for player~1 and player~2, 
respectively, we define the \emph{value} functions
$\va$ and $\vb$ for the players~1 and~2, respectively, as the following 
functions from the state space $S$ to the set $\reals$ of reals:
for all states $s\in S$, let
\[
\va^G(f)(s)  = 
\displaystyle \sup_{\straa \in \Straa} \inf_{\strab \in \Strab} 
\Exp_s^{\straa,\strab}[f]; \quad
\vb^G(-f)(s)  =   
\displaystyle \sup_{\strab\in \Strab} 
\inf_{\straa \in \Straa} \Exp_s^{\straa,\strab}[-f].
\] 
In other words, the values $\va^G(f)(s)$ give the maximal expectation with which player~1 can 
achieve her objective $f$ from state~$s$,
and analogously for player~2.
The strategies that achieve the values are called optimal:
a strategy $\straa$ for player~1  is \emph{optimal} from the state
$s$ for the objective $f$ if 
$\va^G(f)(s)=\inf_{\strab \in \Strab} \Exp_s^{\straa,\strab}[f]$.
The optimal strategies for player~2 are defined analogously.
We now state the classical memoryless determinacy results for $2\half$-player 
games with discounted and long-run average objectives.

\begin{theorem}[Quantitative determinacy~\cite{FilarVrieze97,LigLipp69}] 
\label{thrm:quan-det}
For all $2\half$-player game graphs $G$, the
following assertions hold.
\begin{itemize}
\item
 For all reward functions $r:S \to \rgz$, for all $0<\beta<1$, and all states~$s\in S$, 
 we have 
 \[
  \va^G(\Disc(\beta,r))(s) + \vb^G(\Disc(\beta,-r))(s) =0; 
\]
\[
  \va^G(\LimAvg(r))(s) + \vb^G(\LimAvg(-r))(s) =0.
 \] 
\item Pure memoryless optimal strategies exist for both players from all 
states for discounted and long-run average objectives.
\end{itemize}
\end{theorem}

We now present the definition of the \emph{predecessor} operator $\Pre$.
The operator $\Pre$ is an important operator that is used in many classical
algorithms to solve $2\half$-player games with discounted and long-run 
average objectives.

\begin{definition}[The predecessor operator (Pre)]
\label{def:pre}
Given a game graph $G=((S, E), (\SA,\SB,\SR),\trans)$, the predecessor operator 
$\Pre$ takes a valuation 
$v {:  } S \to\rgz $ and returns a valuation $\Pre(v) {: } S \to \rgz$ 
defined as follows: for every state $s \in S$ we have
\[
\Pre(v)(s)=
\begin{cases}
 \max_{t \in E(s)} v(t)  &  s \in S_1  \\
 \min_{t \in E(s)} v(t)  &  s \in S_2  \\
 \sum_{t \in S} \trans(s,t) \cdot v(t) & s \in S_p.
\end{cases}
\]
\end{definition}

\section{MLA 
for Discounted Objectives}\label{sec-discounted}
In this section we present algorithmic solutions for $2\half$-player games 
and MDPs with discounted objectives.

\smallskip\noindent{\bf Classical Algorithms.}
We present the algorithms to solve a turn-based stochastic games 
with discounted objectives.

\begin{theorem}[\cite{FilarVrieze97,Bertsekas95}] 
\label{thrm-classical}
Given a turn-based stochastic game graph $G$, with a reward function 
$r : S \to \rgz$ and a discount factor $0<\beta<1$, the following assertions hold.
\begin{compactenum}
\item \emph{(Value iteration).} 
Consider the sequence of valuations $v_0, v_1, v_2, \ldots$
as follows: let $v_0 = \mathbf{0}$ and for all $i\geq 0$ and 
$s \in S$ we have 
$$v_{i+1}(s) = (1-\beta)\cdot r(s) + \beta\cdot \Pre(v_i)(s).$$
The sequence $(v_i)_{i\geq 0}$ converges monotonically to 
$ \va^G(\Disc(\beta,r))$.

\item  \emph{(Fixpoint solution).} There exists a valuation $v^*$ 
that is the unique fixpoint of the function 
$f(v)(s) = (1-\beta) \cdot r(s) + \beta\cdot\Pre(v)(s)$, i.e.,
for all $s \in S$ we have  
\[
 v^*(s) = (1-\beta)\cdot r(s) + \beta\cdot \Pre(v^*)(s) 
\]   
and we have  $v^* = \va^G(\Disc(\beta,r))$. 
\end{compactenum}
\end{theorem}

\smallskip\noindent{\em The classical algorithms.} The 
classical algorithms for solving turn-based stochastic games are
based on the result of Theorem~\ref{thrm-classical} and are
as follows.
\begin{compactenum}
\item We obtain the sequence of valuations $(v_i)_{i\geq 0}$ as
given by Theorem~\ref{thrm-classical} by iterating over the 
valuations, and the sequence converges (w.r.t. an error tolerance 
$\vefloat$) to the desired value of
the game.

\item The fixpoint $v^*$ that gives the desired value of the game 
can be obtained by solving optimization problems: if the game
graph is an MDP, then it can be obtained from the solution of a
linear-programming problem~\cite{Manne60}, and for general turn-based 
stochastic games it can be obtained as a solution of a quadratic
programming problem~\cite{Rothblum78}.
\end{compactenum}

\smallskip\noindent{\em Abstract properties for upper and lower bound of 
value functions.}
We first present certain abstract properties of functions that can be
used to obtain upper and lower bounds on the value of stochastic 
game with a discounted objectives.
Later we will present a concrete functions that satisfies
the abstract properties and can be implemented by the magnifying lens 
abstraction techniques.

\begin{theorem}\label{thrm-basic-property}
Let $G$ be turn-based stochastic game graph with 
reward function $r : S \to \rgz$ and discount factor $\beta$.
Let $M=\max_{s \in S} |r(s)|$ and let $Q=\frac{M}{1-\beta}$. 
Consider the  function $f$ on valuations such that 
\[
f(v)(s) = (1-\beta)\cdot r(s) + \beta\cdot \Pre(v)(s). 
\]
Let $f_+$ and $f_{-}$ be two functions on valuations 
that satisfy the following conditions:
\begin{compactenum}
\item $f_+$ and $f_{-}$ are monotonic;
\item for all valuations $v$ we have 
$f_{-}(v) \leq f(v) \leq f_+(v)$;

\item for all valuations bounded by $Q$ (i.e., for all $s \in S$ we 
have $-Q \leq v(s) \leq Q$) we have 
$-Q \leq f_{-}(v) \leq f_+(v) \leq Q$. 
\end{compactenum}
Then there exist least fixpoints $v^*_+$ and $v^*_{-}$ of $f_+$ and $f_{-}$ and
$v^*_{-} \leq \va^G(\Disc(\beta,r)) \leq v^*_+$.
\end{theorem}

In the following  we will use the magnifying lens abstraction
techniques to define functions $f_+$ and $f_{-}$ that satisfies the
properties of the above theorem. 
This will allow us to obtain efficient solution of turn-based stochastic 
games with abstraction techniques.

\smallskip\noindent{\bf Magnifying Lens Abstraction Algorithm.}
{\em Magnifying-lens abstractions\/} (MLA) is a semi-abstract technique that
can cluster states based on value only, disregarding differences in their transition relation.
Let $v^*$ be the discounted sum valuation over $S$ that is to be computed.
Given a desired accuracy $\veabs {>} 0$, MLA
computes upper and lower bounds for $v^*$, spaced less than $\veabs$ apart.
\begin{algorithm}[htb]
\caption{MLA$(G, \beta, r, \veabs, \vefloat)$ Magnifying-Lens Abstraction}
\label{algo-mla}
\begin{tabbing}
{\bf Input} : game $G$, discount factor $\beta$, \\
\qquad \quad \ reward function $r : S \to \rgz$, \\
\qquad \quad \ errors $\veabs > 0$, $\vefloat \ge 0$\\
{\bf Output} :  final partition $R$, valuations $u^+, u^- : R \to\rgz$\\
1.  \ \= $R {:=}$ some initial partition. \\
2.  \>  $u^- {:=} \mathbf{0}$; $u^+ {:=} \mathbf{0}$\\
3.  \> {\bf loop} \\
4.  \> \quad   $u^+$ {:=} $u^-$\\
5.  \> \quad $u^+ {:=}$ GlobalValIter$(G,R,u^+, \beta, r, \max, \vefloat)$\\
6. \> \quad  $u^- {:=}$ GlobalValIter$(G,R,u^-, \beta, r, \min, \vefloat)$\\  
7. \> \quad {\bf if} $||u^+ - u^-|| \geq \veabs$ \\
8. \> \quad \quad {\bf then} $R, u^-, u^+ {:=}$ SplitRegions$(R,u^-,u^+,\veabs)$ \\
9. \> \quad \quad {\bf else} {\bf return} $R, u^-, u^+$ \\
10. \> \quad {\bf end if} \\
11. \> {\bf end loop}
\end{tabbing}
\vspace*{-2ex}
\end{algorithm}

\smallskip\noindent{\em Algorithm Sketch.}
The MLA algorithm is shown in Algorithm~\ref{algo-mla}.
The algorithm has parameters $G$, $\beta$, $r$,  and errors $\veabs>0$, $\vefloat \ge 0$.
Parameter $\veabs$ indicates the allowed maximum difference between the lower
and upper bounds returned by MLA. 
MLA starts from an initial partition (set of regions) $R$ of $S$.
The initial partition $R$ is obtained either from the user or from the property. 
Statement 2 initializes the valuations $u^-$ and $u^+$ to  $\mathbf{0}$
since discounted sums are computed as least fixpoints. 
 MLA computes the lower and upper bounds as valuations 
$u^-$ and $u^+$ over $R$ by GlobalValIter Algorithm (Algorithm~\ref{algo-global}). 
Global iterations, when implemented as a value iteration (Algorithm~\ref{algo-global}), contains an extra parameter $\vefloat {>} 0$. 
Parameter $\vefloat$, stopping parameter of classical value iteration, specifies the degree of precision to which the global value iteration should converge.
For accurate global iterations, we can set the parameter $\vefloat$ to $0$.
The partition is refined, until the difference between $u^-$ and
$u^+$, for all regions, is below a specified threshold. 
%

\smallskip\noindent{\em Global Value Iteration (GlobalValIter).} 
To compute $u^-$ (resp. $u^+$), GlobalValIter considers each region $x \in R$ in turn, 
and performs a {\em magnified iteration (MI):\/} it improves the
bounds $u^-(x)$ (resp. $u^+(x)$) by solving the sub-games on the concrete states in $r$.
\begin{algorithm}[htb]
\caption{GlobalValIter$(G, R, u, \beta, r, h, \vefloat)$ Global Value Iteration
}
\label{algo-global}
\begin{tabbing}
{\bf Input} : game $G$, partition $R$, valuation $u : R \to \rgz$, \\
\qquad \quad  \ discount factor $\beta$, reward function $r : S \to \rgz$, \\
\qquad \quad \ $h \in \set{\max,\min}$, error $\vefloat \ge 0$\\
{\bf Output} :  valuation $u : R \to \rgz$\\
1.  \=  {\bf repeat} \\
2.  \>  \quad $\hat{u} {:=} u$\\
3.  \> \quad  {\bf for} $x \in R$ {\bf do}\\
4.  \> \quad \quad $u(x) {:=}$ MagIter$(G,R,x,\hat{u},\beta,r,h,\vefloat)$\\
5.  \> \quad {\bf end for} \\
6. \> {\bf until} $||u - \hat{u}|| \leq \vefloat$ \\
7. \> {\bf return} u
\end{tabbing}
\vspace*{-2ex}
\end{algorithm}

\smallskip\noindent{\em MagnifiedIteration (MagIter).} 
The goal of the magnified iteration algorithm is to 
either (a) iterate function $f_+$ and $f_{-}$ with properties
of Theorem~\ref{thrm-basic-property} or (b) obtain 
fixpoints of $f_+$ and $f_{-}$. 
To obtain the desired functions we define an auxiliary function
$g$ and a magnified predecessor operator $\MPre$ and then
present the magnifying lens abstraction implementation of 
$\MPre$.


\begin{definition}
Given a game graph $G=((S, E), (\SA,\SB,\SR),\trans)$, two states $s,t \in S$, 
a partition $R$, a valuation $v : S \to \rgz$, $h \in \set{\max, \min}$, 
we define the following auxiliary function $g$ as follows:
\[
g(s,h,R,v)(t) = \ \left\{
\begin{array}{ll}
 {\displaystyle  v(t)} &  t \in [s]_R\\
 {\displaystyle  h\{v(t') \mid t' \in [t]_R\}} &  t \not\in [s]_R\\
\end{array}\right.
\] 
\end{definition}

The function $g$ is as follows: given two states $s$ and $t$, 
a valuation $v$, a partition $R$ and a function $h\in \set{\max,\min}$,
it returns the valuation $v(t)$ if $s$ and $t$ belong to the same partition,
otherwise it returns the result of applying $h$ to the values $v(t')$ 
of the states $t'$ that belongs to the same region as $t$.
We now define the magnified predecessor operator $\MPre$ that is similar to 
$\Pre$ but applies the function $g$ to obtain values.

\begin{definition}[Magnified Predecessor Operator ($\MPre$)]\label{defi-mpre}
Given a game graph $G=((S, E), (\SA,\SB,\SR),\trans)$, a partition $R$,  
a valuation  $v : S \to \rgz$, $h \in \set{\max,\min}$, 
we define the valuation $\MPre(h,v,R) : S \to \rgz$ as follows: 
let $z$ represent  $(s,h,R,v)$, then for  all states $s \in S$, we have
\[
 \MPre(h,v,R)(s) =  \ \left\{ 
\begin{array}{ll}
  {\displaystyle  \max_{t \in E(s)} g(z)(t)} &  s \in S_1 \\
   {\displaystyle  \min_{t \in E(s)} g(z)(t)}  &  s \in S_2\\
   {\displaystyle \sum_{t \in E(s)} \trans(s,t) \cdot g(z)(t)} & s \in S_p
\end{array}\right.
\]
\end{definition}

\begin{lemma}[Properties of $\MPre$]\label{lemm-1-app}
Given a game graph $G$, for all partitions $R$ and all 
$h \in \set{\max,\min}$, we have
\begin{compactenum}
\item $\MPre(h,v,R)$ is monotonic i.e. for two valuations $v, v'$, if 
$v \le v'$, then $\MPre(h,v,R) \le \MPre(h,v',R)$.
\item If valuation $v$ is bounded by $Q$, then $\MPre(h,v,R)$ is also bounded 
 by $Q$.
\item If $h$ is $\max$, then $\Pre(v) \le \MPre(h,v,R)$, and if $h$ is $\min$, 
then $\Pre(v) \ge \MPre(h,v,R)$.
\end{compactenum} 
\end{lemma}
The above lemma shows that $\MPre$ with $h$ as $\max$ and $\min$, respectively, 
satisfies all the properties of $f_{+}$ and $f_{-}$ of Theorem~\ref{thrm-basic-property},
respectively. 
Hence we obtain the following lemma.

\begin{lemma}\label{lemm-2-app}
Given a game $G$, for all partitions $R$, and all valuations $v: S \to \rgz$, 
consider the following functions: 
\[
 l_+(v)(s) = \beta\cdot r(s) + (1-\beta)\cdot \MPre(\max,v,R)(s); 
\]
\[
 l_-(v)(s) = \beta\cdot r(s) + (1-\beta)\cdot \MPre(\min,v,R)(s). 
\]
Then there exist least fixpoints $v^*_+$ and $v^*_-$ of $l_+$ and $l_-$, respectively,
such that $v^*_- \le \va^G(\Disc(\beta,r)) \le v^*_+$.
\end{lemma}


%
\smallskip\noindent{\em Magnified Iteration Implementation.}
We now present the implementation details of the magnified iteration techniques. 
The operator $\MPre$ takes as input a valuation over the whole state-space $S$,
and returns a valuation over the whole state space.
In the magnifying lens abstraction implementation, our goal is to save space,
and operate on valuations that are not on the whole state space.
To achieve this goal, for a given region $x \in R$, we define a new operator $\MPrex$ and present its relation with 
$\MPre$.

\begin{definition}[$\MPrex$]\label{defi-mprex}
Given a game graph $G=((S, E), (\SA,\SB,\SR),\trans)$, a partition $R$,  
a region $x \in R$,  valuations  $u : R \to\rgz$, $v_x : x \to \rgz$,  
we define the valuation $\MPrex(v_x,R,u) : x \to \rgz$ as follows: 
for  all states $s \in x$, we have
\[
 \MPrex(v_x,R,u)(s) =  \ \left\{ 
\begin{array}{ll}
  {\displaystyle  \max_{t \in E(s)} \hat{g}(y)(t)} &  s \in S_1 \\
   {\displaystyle  \min_{t \in E(s)} \hat{g}(y)(t)}  &  s \in S_2\\
   {\displaystyle \sum_{t \in E(s)} \trans(s,t) \cdot \hat{g}(y)(t)} & s \in S_p
\end{array}\right.
\]
where $y$ represents  $(s,R,v_x,u)$. The auxiliary function $\hat{g}$ can be defined as follows:
\[
\hat{g}(s,R,v_x,u)(t) = \ \left\{
\begin{array}{ll}
  v_x(t) &  t \in [s]_R\\
 u([t]_R) &  t \not\in [s]_R\\
\end{array}\right.
\]
\end{definition}

Observe that $\MPrex$  
takes a valuation on the states of a region $x$ (instead of a valuation on the whole 
state space),
and a valuation on the partition of the state space (and hence
requires much smaller memory than a valuation on the whole 
state space).
The following lemma establishes the relation of $\MPre$ and
$\MPrex$.
Hence we always achieve the implementation of the $\MPre$ 
operator as $\MPrex$.

\begin{lemma}[Relation of $\MPre$ and $\MPrex$]\label{lemm-3-app}
Given a game graph $G$, for all partitions $R$ and all 
$h \in \set{\max,\min}$, for all valuations $v:S \to \rgz$, 
for all $x \in R$, let $v_x: x \to \rgz$ be a valuation such 
that $v_x(s)=v(s)$ for all $s \in x$, and let $u: R\to \rgz$
be a valuation such that $u(x)= h\set{v(s) \mid s \in x}$.
Then we have $\MPrex(v_x,R,u)(s) = \MPre(h,v,R)(s)$ for 
all $s \in x$.
\end{lemma}

Magnified iteration, which involves the $\MPrex$ implementation of $\MPre$ 
using magnifying-lens abstraction technique, can be done in two ways 
like the classical algorithms.
We present them below.

\smallskip\noindent{\em Solution of fixpoint by optimization.} The fixpoints
of the functions that provide upper and lower bound on the value 
using $\MPre$ and $h$ as $\max$ and $\min$ can be obtained by solution of 
optimization problems.
We present the fixpoint solution for the case when $h$ is $\max$ and the case when
$h$ is $\min$ is similar.
Given a partition $R$, we have two valuation variables $u^+: R \to \rgz$
and $v: S \to \reals$ and we denote by $v_x$ the valuation variable $v$ 
restricted to a region $x \in R$.
We have a set of global constraints that specifies that in 
every region $x$ the value $u^+(x)$ is the maximum value of $v_x(s)$ for all 
$s \in x$; i.e., we have the following constraints
\[
u^k(x) = h_{s \in x} v_x (s)  \textrm{ for all } x\in R.
\]
Along with the above constraints we have local constraints for every region 
$x \in R$ and it specifies that $v_x(s)$ should satisfy the fixpoint 
constraints for $\MPrex$.
In other words, for every region $x \in R$ we have the following set of 
local constraints:
\[
v_x(s) =  (1-\beta)\cdot r(s) + \beta\cdot \MPrex (v_x,R,u^+)(s) \textrm{ for all } s \in x. 
\]
Thus instead of solving one huge optimization problem, using the 
$\MPrex$ we decompose the optimization problem into many smaller sub-problems
with independent sub-parts.
Thus the solution is more space efficient and can be achieved faster in practice.
Also notice that the solution by optimization to obtain the fixpoint 
correspond to the solution of magnified iteration (MagIter) with 
$\vefloat=0$.

\begin{theorem}[Correctness of Approximation]
Given a turn-based stochastic game $G =((S, E), (\SA,\SB,\SR),\trans)$, a discount factor $\beta$, 
a reward function $r$, and error bounds $\veabs {>} 0$, and $\vefloat=0$, the following assertions hold:
let $(R,u^+,u^-)=\mathit{MLA}(G,\beta,r,\veabs,0)$, then 
\begin{compactenum}
\item for all $s \in S$ we have 
$u^-([s]_R) \leq \va^G(\Disc(\beta,r))(s) \leq u^+([s]_R)$ ; and 
\item for all $x \in R$ we have $u^+(x)-u^-(x) \leq \veabs$.
\end{compactenum} 
\end{theorem}

\noindent{\em Value iteration implementation of MagIter.} The  Magnified Iteration (MagIter) 
step can also be implemented as a value iteration approach. 
When MagIter is implemented as a value iteration, then we 
require that $\vefloat {>} 0$. 
The parameter $\vefloat$  specifies the degree of precision to which 
the local, magnified value iteration should converge. 
Algorithm~\ref{algo-mi-vi} describes the formal description of the procedure.

\begin{algorithm}[htb]
\caption{MagIter$(G,R,x,u,\beta,r,h,\vefloat)$} 
  \label{algo-mi-vi}
  
\begin{tabbing}
{\bf Input} : game $G$, partition $R$, a region $x \in R$, \\
\qquad \quad  \ valuation $u : R \to \rgz$, discount factor $\beta$, \\
\qquad \quad \  reward function $r: S \to \rgz$\\
\qquad \quad \  $h \in \{max, min\}$, error $\vefloat$\\
{\bf Output} :  a value $u(x) : \rgz$ \\
{\bf Data Structure} : $v, \hat{v}$: valuations over $x$ \\
1.  \= {\bf for} $s \in x$ {\bf do} $v(s) {=} u(r)$ {\bf end for}\\
2. \> {\bf repeat}\\ 
3. \quad $\hat{v} {:=}  v $\\
3. \> \quad {\bf for} $s \in x$ {\bf do} \\
4. \>  \quad \quad $v(s) {=} (1-\beta)\cdot r(s) + \beta\cdot \MPrex (\hat{v},R,u)(s)$ \\
5. \> \quad {\bf end for}\\
6. \> {\bf until} $||v - \hat{v} || \leq \vefloat$ \\  
7.  \> {\bf return} $h \set{v(s) \mid s \in x}$
\end{tabbing}
\vspace*{-2ex}
\end{algorithm}

\begin{theorem}[Termination and Correctness]
Given a turn-based stochastic game $G =((S, E), (\SA,\SB,\SR),\trans)$, a discount factor $\beta$, 
a reward function $r : S \to \rgz$, for all error bounds $\veabs {>} 0$,  the following assertions hold.
\begin{compactenum}
\item For all $\vefloat>0$, the call $\mathit{MLA}(G,\beta,r,\veabs,\vefloat)$ terminates.
\item There exists an error bound $\vefloat$ such that if 
$(R,u^+,u^-)=\mathit{MLA}(G,\beta,r,\veabs,\vefloat)$, then 
\begin{compactenum}
\item for all 
$s \in S$ we have $u^-([s]_R) \leq \va^G(\Disc(\beta,r))(s) \leq u^+([s]_R)$; 
and 
\item for all $x \in R$ we have $u^+(x)-u^-(x) \leq \veabs$. 
\end{compactenum}
\end{compactenum}
\end{theorem}

\noindent{\em Adaptive refinement step (SplitRegions).}
The step {\em SplitRegions} is obtained by adaptive refinement of regions 
with 
large imprecisions. 
We denote the {\em imprecision\/} of a region $x$ by $\Delta(x) = u^+(x) - u^-(x)$. 
MLA adaptively refines a partition $R$ by splitting all regions $x$
having $\Delta(x) > \veabs$.
The refinement scheme is simple and easy to implement. 
Thus a call to $SplitRegions(R,u^+,u^-,\veabs)$ returns a triple 
$\tilde{R}, \tilde{u}^-, \tilde{u}^+$, consisting of the new
partition with its upper and lower bounds for the valuation. 
Like ~\cite{deAlR07}, we also tried other refinement heuristics, but none of them
gave strictly better results. 

\paragraph{Space Savings}
For value iteration algorithm, the space requirement is equal to  the size of state-space $|S|$,
the domain of $v$. 
For MLA, the space requirement is equal to be the maximum value of 
$2\cdot |R| + \max_{x \in R} |x|$.
The expression gives the maximum space
required to store the valuations $u^+$, $u^-$, as well as the values
$v$ for the largest magnified region.
Since $\max_{x\in R} |x| \geq (|S|/|R|)$, the space complexity of
the algorithm is (lower) bounded by a square-root function $ \sqrt{8\cdot|S|}$.
However, this bound is provided for the concrete implementation. 

\section{MLA for 
Long-run Average Objectives}\label{sec-longrun}

In this section we present magnifying lens abstraction solution 
for a class of stochastic games with long-run average objectives. 
We first describe the efficient classical solution and 
then present our magnifying lens abstraction solution.

\smallskip\noindent{\em Value iteration for long-run average objectives.} 
A \emph{value iteration} algorithm can be used to compute the long-run 
average value as follows: for a state $s$ we compute by value iteration 
the maximum expected sum of the rewards for $k$-step starting from 
$s$, and we denote this sum as $S(k,s)$. 
Then the value of the state $s$ is $\lim_{k \to \infty} \frac{S(k,s)}{k}$.
However, this technique is not very practical as  $S(k,s) \to \infty$ 
and $S(k,s)$ diverges fast towards infinity.
Hence computing $S(k,s)$ and dividing by $k$ is computationally 
expensive and not very practical.
This problem can be alleviated by {\em relative value iteration} algorithm that 
subtracts a number $c \in \reals$ in each iteration.
This technique trims the values for all states simultaneously, and this technique 
is an efficient way to compute values in games that have same values in 
all states.
\begin{lemma}
\label{lemm-rel-val} 
Consider a turn-based stochastic game graph $G=((S, E), (\SA,\SB,\SR),\trans)$ 
with a reward function $r:S\to \rgz$.
For a real number $c$, consider a sequence of valuations $(v_i)_{i \ge 0}$ as follows: 
let $v_0 = \mathbf{c}$ and for all $i \ge 0$ and $s \in S$ we have,
$v_{i+1}(s) =  r(s) -c + \Pre(v_i)(s).$
If there exists a real value $v^*$ such that for all $s \in S$ we have 
$\va^G(\LimAvg(r))(s)=v^*$, then the following conditions hold:
\begin{compactenum}
\item The sequence $(v_{i})_{i \ge 0}$ diverges to $+\infty$ iff $c < v^*$.
\item The sequence $(v_{i})_{i \ge 0}$ diverges to $-\infty$ iff $c > v^*$.
\end{compactenum}
\end{lemma}
The relative value iteration algorithm chooses a real value $c$, and 
then adjusts the value of $c$ adaptively depending on whether the sequence 
$(v_i)_{i \ge 0}$, given the chosen value $c$, diverges to $+\infty$ or $-\infty$, 
otherwise the chosen real number $c$ is value of the game.

\smallskip\noindent{\bf MLA 
for Stochastic Games.} 
We develop magnifying lens abstraction solution for stochastic games 
under the assumption that there is a \emph{uniform value} $v^*$ such that 
every state has the same value $v^*$. 
Later we will consider the question of presenting criteria for its existence.
For MDPs we will present our solution of all MDPs (without the assumption 
of existence of uniform value).
The magnifying lens abstraction solution for stochastic games with long-run 
average objective is based on the following lemma.
\begin{lemma}[Magnified Relative Value Iteration.]\label{lemm-mag-rel-val}
Given a game graph $G$, for all partitions $R$,  consider two sequence of valuations  
$(v^+_i)_{i \ge 0}$ and $(v^-_i)_{i \ge 0}$ as follows : 
let $v^+_0 = v^-_0 = c$ and for all $i \ge 0$ and $s \in S$ we have:
\[
  v^+_{i+1}(s) = r(s) - c + \MPre(\max,v^+_{i},R)(s)
\]
\[
  v^-_{i+1} (s) = r(s) - c + \MPre(\min,v^-_{i},R)(s)
\] 
If there exists a real value $v^*$ such that for all $s \in S$ we have $\va^G(\LimAvg(r))(s)=v^*$, then 
the following conditions hold:
\begin{compactenum}
\item If the sequence $(v^{-}_{i})_{i \ge 0}$ diverges to $+\infty$, then $c < v^*$.
\item If the sequence $(v^{+}_{i})_{i \ge 0}$ diverges to $-\infty$, then $c > v^*$.
\end{compactenum}
\end{lemma}

Lemma~\ref{lemm-rel-val} relates the divergence of the sequence $(v_i)_{i\ge 0}$ for 
a chosen $c$ and the value of the game in both directions (iff conditions giving 
necessary and sufficient conditions), whereas the Lemma~\ref{lemm-mag-rel-val} 
(with magnified pre operator) relates the value of $c$ and the divergence of 
the sequence in one direction, i.e., if the sequence diverges in a given 
direction (i.e,$(v^+_i)_{i\geq 0}$ to $-\infty$ or $(v^{-}_i)_{i \ge 0}$ to $+\infty$),
then we conclude the relation of the value of the game and the chosen value $c$.
We now present the algorithm that is based on Lemma~\ref{lemm-mag-rel-val}.

\smallskip\noindent{\em Algorithm Sketch.} 
Algorithm~\ref{algo-long-game} provides an algorithm to approximate the long-run average value of 
a stochastic game $G$ such that every state has the same value.
Algorithm uses Lemma~\ref{lemm-mag-rel-val} to obtain upper and lower bound on the value of the 
game by a dichotomic (binary) search.
The search space in bounded by the interval $[c^-,c^+]$, where $c^+$ and $c^-$ denote an upper and 
a lower bound on the value of the game, respectively.
The initial value of $c^+$ (resp. $c^-$) is obtained from the maximum (resp. minimum) reward value
of the game.
Each iteration of this binary search starts by setting $c$ to the mid-point of the interval. 
If with the chosen value of $c$, the sequence $(v_{i}^+)_{i \geq 0}$ diverges to 
$-\infty$, then $c$ is an upper bound on the value of the game, and $c^+$ is set 
(decreased) to $c$.
Similarly, if with the chosen value of $c$, the sequence $(v_{i}^-)_{i \geq 0}$ diverges
to $-\infty$, then $c$ is an lower bound on the value of the game and 
$c^-$ is set (increased) to $c$.
\begin{algorithm}[htb]
\caption{MLALongRun$(G, r, \veabs, k)$ Magnifying-Lens Abstraction}
\label{algo-long-game}
\begin{tabbing}
{\bf Input} : game $G$, \ reward function $r : S \to \rgz$ \\
\qquad \quad errors $\veabs > 0$,\\ 
\qquad \quad maximum number of iterations $k : integer$\\
\qquad \quad a real value $ratio \in [0,1]$  \\
{\bf Output} :  final partition $R$\\
1.  \ \= $R {:=}$ some initial partition\\
2.  \> $c^+ {:=} \max_{s \in S} r(s)$,  $c^- {:=} \min_{s \in S} r(s)$\\
3.  \> {\bf while} $(c^+ - c^-) \le \veabs$ {\bf do}\\
4.  \> \quad $c { := } {(c^{+} + c^{-})/2}$\\
5.  \>  \quad $d^+, v^+ {:=} $\ CheckDivergence($G,r,R,c,\max, k$)\\
6.  \>  \quad $d^-, v^-  {:=} $\ CheckDivergence($G,r,R,c,\min, k$)\\
7.  \>  \quad{\bf if} $d^{-} = +$ {\bf then}  $c^- {:=} c$\\  
8.  \>  \quad {\bf else if} $d^{+} = -$ {\bf then}  $c^+ {:=} c$\\  
8.  \> \quad {\bf else}  $R$ := SplitRegions$(R,v^{+},v^{-}, \veabs, ratio)$ \\
9. \>  \quad {\bf end if} \\
10. \> {\bf end while}
\end{tabbing}
\vspace*{-2ex}
\end{algorithm}
The procedure {\em CheckDivergence} returns a verdict on divergence of the sequence
by computing $(k+1)$ elements of the sequence.
The verdict $+$,$-$  denote the divergence to $+\infty$  and $-\infty$ respectively. 
However, the verdict $?$ tells that either (a)~$k$-elements of the sequence is not enough to 
detect the divergence or (b)~the sequence may not diverge to $+\infty$ or $-\infty$.
If {\em CheckDivergence} returns $+$ for the $d^-$, then $c^-$ is set to $c$, and
if {\em CheckDivergence} returns $-\infty$ for the $d^+$, then $c^+$ is set to $c$.
Otherwise we do not have enough information to update $c^+$ or $c^-$, and algorithm 
refines the partition $R$ by invoking {\em SplitRegions} procedure.
The {\em imprecision} of a region $x \in R$ is denoted by $\Delta(x) = v^+(x) - v^-(x)$. 
The procedure {\em SplitRegions} splits a number (precisely $ratio\cdot |R|$) of high imprecision regions. 

\smallskip\noindent{\em Detecting divergence to $+\infty$ and $-\infty$.}
Algorithm~\ref{algo-checkdiv} illustrates the procedure {\em CheckDivergence} to detect the
divergence of a sequence starting from a given $c$ and a given number of iterations $k$.
If the value for every state increases beyond $c$ after $k$-iterations, then the 
sequence diverges to $+\infty$ (the procedure returns $+$), 
and if the value for every state decreases below $c$, 
then the sequence diverges to $-\infty$ (the procedure returns $+$). 
Otherwise, the divergence to $+\infty$ or $-\infty$ cannot be concluded and 
then the procedure returns the verdict $?$.
The algorithm also returns the ($k+1$)-th valuation of the sequence starting from $c$.
\begin{algorithm}[htb]
\caption{CheckDivergence$(G, r, R, c, h,k)$}
\label{algo-checkdiv}
\begin{tabbing}
{\bf Input} : game $G$, \ reward function $r : S \to \rgz$,\\  
\qquad \quad a partition $R$, a chosen value $c : \rgz$,\\ 
\qquad \quad $h \in \set{\max,\min}$,\\ 
\qquad \quad maximum number of iterations $k : integer$\\
{\bf Output} :  enum $d \in \set{+,-,?}$, valuation $v : R \to \rgz$\\
1.  \ \= $v_0 {:=} \mathbf{c}$, $d {:=} ?$ \\
2.  \> {\bf for} $i=0$ {\bf to} $k$  {\bf do}\\
3.  \>  \quad {\bf for each} $x \in R$\\ 
4.  \> \quad  \quad $v_{i+1}(x)$ := MagIter2$(G,R,x,v_i,r,h.k)$\\
5.  \>  \quad {\bf end for}\\
6.  \> {\bf end for}\\
7.  \>  {\bf if} ($\min_{x \in R} v_{k+1}(x) > c$) {\bf then} $d := +$\\
8.  \>  {\bf if} ($\max_{x \in R} v_{k+1}(x) < c$) {\bf then} $d := -$\\
9.  \> {\bf return} $d, v_{k+1}$
\end{tabbing}
\vspace*{-2ex}
\end{algorithm}

\smallskip\noindent{\em Magnified Iteration (Long Run Average version).} 
Algorithm~\ref{algo-mi-long} provides the details of the magnified
iteration of a region $x \in R$.
The algorithm completes value-iteration for $k$-iterations over the states of the region $x$, and 
summarizes the values to a single value.
Like discounted case, we use $\MPrex$ operator for the magnifying lens abstraction implementation.
\begin{algorithm}[htb]
\caption{MagIter2$(G,R,x,u,r,h,k)$} 
  \label{algo-mi-long}
  
\begin{tabbing}
{\bf Input} : game $G$, partition $R$, a region $x \in R$, \\
\qquad \quad  \ valuation $u : R \to \rgz$, \\
\qquad \quad  reward function $r: S \to \rgz$,
 $h \in \set{\max,\min}$,\\
\qquad \quad    maximum number of iterations $k : integer$\\
{\bf Output} :  a value $u(x) : \rgz$ \\
1.  \= {\bf for} $s \in x$ {\bf do} $v_0(s) {=} u(r)$ {\bf end for}\\
2. \>  {\bf for} $i=0$ {\bf to} $k$  {\bf do}\\ 
3. \> \quad {\bf for} $s \in x$ {\bf do} \\
4. \>  \quad \quad $v_{i+1}(s) { := }  r(s) - c + \MPrex(v_i,R,u)(s)$ \\
5. \> \quad {\bf end for}\\
6. \> {\bf end for}\\
7.  \> {\bf return} $h \set{v_{k+1}(s) \mid s \in x}$
\end{tabbing}
\vspace*{-2ex}
\end{algorithm}

\begin{theorem}[Termination and Correctness]
Let $G =((S, E), (\SA,\SB,\SR),\trans)$ be a turn-based stochastic game with
a reward function $r : S \to \rgz$ such that there exists $v^* \in \rgz$ and for all 
$s \in S$ we have $\va^G(\LimAvg(r))(s)=v^*$.
For all error bounds $\veabs {>} 0$, the following assertions hold.
\begin{compactenum}
\item The call $\mathit{MLALongRun}(G,r,\veabs,k)$ terminates.
\item There exists a positive integer  $k$ such that if 
$(R,u^+,u^-)=\mathit{MLALongRun}(G,\beta,r,\veabs,k)$, then 
\begin{compactenum}
\item for all 
$s \in S$ we have $u^-([s]_R) \leq \va^G(\LimAvg(r))(s) \leq u^+([s]_R)$; 
and 
\item for all $x \in R$ we have $u^+(x)-u^-(x) \leq \veabs$. 
\end{compactenum}
\end{compactenum}
\end{theorem}

\noindent{\em Ensuring uniform value.} 
The following theorem presents a sufficient condition to ensure 
uniform value in a turn-based stochastic game (i.e., the same 
value everywhere). The condition can be checked 
in polynomial time using algorithms for solving turn-based stochastic 
reachability games with qualitative winning criteria~\cite{Condon92}.
For a state $t$ we denote by $\Diamond t$ the set of paths that 
reaches $t$.

\begin{theorem}
Consider a turn-based stochastic game graph $G$ with a reward function $r : S \to \rgz$.
Suppose there exists a state $t$ such that the following conditions hold: 
\begin{compactenum}
\item for all $s \in S$ there exists a player~1 strategy $\straa$ such 
that against all player~2 strategies $\strab$ we have 
$\Prb_s^{\straa,\strab}(\Diamond t)>0$;
and 
\item for all $s \in S$ there exists a player~2 strategy $\strab$ such 
that against all player~1 strategies $\straa$ we have 
$\Prb_s^{\straa,\strab}(\Diamond t)>0$.
\end{compactenum}
Then there exists a real value $v^*$ such that for all $s \in S$ we have 
$\va^G(\LimAvg(r)) =v^*$.
\end{theorem}
\begin{proof}
Suppose condition~1 holds, and then by existence of pure memoryless optimal 
strategies in turn-based stochastic reachability games~\cite{Condon92}, there
is a witness pure memoryless strategy $\straa^*$ to witness that for all 
states $s$ the state $t$ is reached with positive probability against all 
player~2 strategies. 
Hence if we fix any pure memoryless counter strategy $\strab$ for player~2 
the closed recurrent set must contain $t$. 
From the existence of pure memoryless optimal strategies for 
turn-based stochastic games with reachability and safety objectives, 
it follows that player~1 can ensure that from all states $s$ the state $t$ is 
reached with probability~1.
Hence for all states $s$ we have 
$\va^G(\LimAvg(r))(s) \geq \va^G(\LimAvg(r))(t)$.
Similarly, if condition~2 holds, then player~2 can ensure that $t$ can be 
reached with probability~1 from all states $s$ and hence for all states $s$ 
we have  $\va^G(\LimAvg(r))(s) \leq \va^G(\LimAvg(r))(t)$.
Hence $v^*= \va^G(\LimAvg(r))(t)$ is the witness real value to show that the
claim holds.
\qed
\end{proof}

\smallskip\noindent{\bf MLA for MDPs.}
Now we present the magnifying lens abstraction 
solution for all MDPs with long-run average 
objectives (i.e., the solution works for MDPs such that values 
at different states may be different).
The main idea relies on the \emph{end component} decomposition
of an MDP. 

\begin{definition}[End component] 
Given an MDP $G =((S, E), (\SA,\SR),\trans)$ a set $C$ of states
is an \emph{end component} if the following conditions hold:
(a)~the set $C$ is strongly connected component in the graph induced
by $(S,E)$; and 
(b)~for all probabilistic states $s \in C \cap \SR$, all out-going edges
of $s$ is contained in $C$, i.e., $E(s) \subseteq C$. 
An end component $C$ is \emph{maximal} if for any end component $C'$ we 
have  either 
(a)~$C' \subseteq C$ or (b)~$C' \cap C =\emptyset$. 
\end{definition}

The following theorem states that given an MDP with a long-run 
average objective, if we consider the sub-game 
graph induced by an end component $C$, then all states in $C$ would
have the same value.

\newcommand{\subgraph}{\upharpoonright}
\begin{theorem}[Existence of uniform value for MDPs]
\label{theo-end-comp-uniform} 
Let $G =((S, E), (\SA,\SR),\trans)$ be an MDP with a reward function 
$r$. 
Consider an end component $C$ in $G$ and the sub-game graph 
$G \subgraph C$ induced by $C$.
Then there exists a real value $v^*$ such that for all $s \in C$ 
we have $\va^{G\subgraph C}(\LimAvg(r))(s) =v^*$.
\end{theorem}

It follows from Theorem~\ref{theo-end-comp-uniform} that if we consider the 
sub-game graph induced by an end component of an MDP, then the condition 
of uniform value (all states having the same value) is satisfied. 
It follows from the results of~\cite{luca-thesis,CY95} that in an MDP 
for all strategies with probability~1 the set of states visited infinitely
often is an end component. 
Hence a pure memoryless optimal strategy consists in reaching the correct 
end component, and then play optimally in the end component.
Thus we obtain the following magnifying lens abstraction algorithm for MDPs; 
the algorithm consists of the following steps:
\begin{compactenum}
\item the MDP is decomposed into its maximal end components (this can be 
achieved in quadratic time);
\item in the sub-game graph induced by an maximal end component we 
approximate the values by the general algorithm 
(since the uniform value condition satisfied we can apply the general algorithm 
for stochastic games);

\item once the values in every maximal end component are approximated 
we can collapse every maximal end component as a single state and obtain 
an MDP with no non-trivial end component (every end component is a single state end component), 
and then compute the values by an algorithm that computes the maximal value that 
can be reached in an MDP with no non-trivial end components. 
 
\end{compactenum}

\vspace{-1em}
\section{Examples and Experimental Results}\label{sec-ex}
\vspace{-1em}
In this section, we provide examples and case studies on MDP models with large state-spaces and the locality property. 
In our implementation we only consider MDPs because common probabilistic model-checkers (like PRISM) only support MDPs. 
In future work we will consider the stochastic games implementations. Our examples show that the value-based abstraction 
methods perform better than the transition-based abstraction methods for MDPs with locality property, although for 
some other case studies the algorithm may provide worse performance. We first present the examples, 
then our symbolic implementation of MLA algorithms for MDPs with discounted objectives and finally, our experimental results.

\begin{figure*}[htb]
\centering
{\footnotesize
\begin{tabular}{| c | l |r| r | r r  | r r r |}\hline
Example & \multicolumn{1}{|c|}{Parameters} & \multicolumn{1}{|c|}{States} &  \multicolumn{1}{|c|}{Transitions} &\multicolumn{2}{|c|}{non-MLA} & \multicolumn{3}{|c|}{MLA}  \\ \cline{5-9}
& & &  & \multicolumn{1}{c}{Nodes}  &\multicolumn{1}{c}{Time} & \multicolumn{1}{|c}{Nodes} & \multicolumn{1}{|c}{Time} & \multicolumn{1}{c|}{Regions}  \\ \hline
Planning  & n=256,m=40 & 65,537 & 265,419  & 3,981 &  28 &  1,658  &   37  & 330  \\ 
           	& n=512,m=40  & 262,145 & 1,063,603 & 12,420 &  106   &  3,324 &  191  & 1,121  \\
          	& n=1024,m=50 & 1,048,577& 4,211,564   &  15,365  & 616 & 4,596 & 883 & 1,670 \\ 
          \hline
Auto & $n_{\max}$=2,047 $t_{\max}$=2,047 & 4,194,304 & 20,965,376 &  12,719 & 65 &  1,171 & 209  & 99 \\
 Inventory               & $n_{\max}$=2,047 $t_{\max}$=4,095 & 8,388,608 & 41,930,752 &  12,676 & 114 &  1,095 & 259  & 99 \\
 	           	   & $n_{\max}$=4,095,$t_{\max}$=4,095 & 16,777,216 & 83,873,792 & 25,434 & 287  & 6,293 &  606 &  99\\   
	                         \hline
	                         Machine  & n =1023, tm=1023 &  1,047,552 &  3,141,632 & 6,051  & 17 & 419 & 64  & 63\\
Replacement    	  & n =2047, tm=2047 & 4,192,256 &  12,574,720  &  12,185 & 47 & 960 & 152  & 64\\
	                             & n =4095, tm=4095 &  16,773,120  &  50,315,264  &  24,461 & 141 & 960 & 364  & 65\\
               \hline

Network & M=7,tm =2,047  &  1,781,760 & 7,157,760  &  328 & 5  & 227 & 11  & 245 \\
Protocol     & M = 15, tm=2,047 &  7,745,536 & 31,045,632 & 369 & 8  & 267 & 24 & 647\\
	         & M = 15, tm=4,095 & 15,491,072 &  62,091,264 & 369 & 17  & 267 & 33 & 647\\
                \hline                
                
\end{tabular}
\caption{Experimental results: Symbolic discounted MLA, compared to  discounted value iteration}
\label{results}
}
\end{figure*}

{\bf Planning:} We consider an MDP that models the movement of a robot in a two-dimensional ($n \times n$) grid. The grid contains $m$ mines. The robot at position $(x,y)$ can choose to move in any of forward, backward,  left or right direction to reach the positions ($x+1,y$),  ($x-1,y$), ($x,y-1$) or ($x,y+1$) respectively. However, there is a chance $p$ that the robot may not reach the desired positions and {\em dies} due to the explosion of a mine, and in that case the robot reaches a special state called {\em sink state}. The probability distribution (i.e., $p$) is a function over distances from the $m$ mines. The robot spends power for its movement and collects rewards (i.e. recharges) associated with the chargeable states when it visits them. In this example, the state-space is two-dimensional, and every state has transitions to the next states in the state-space (i.e., has the locality property). The robot needs to explore the grid points in an intelligent manner such that the robot spends minimum energy. The property of interest is either to maximize the {\em discounted reward} or the {\em long run-average} reward of the robot.

{\bf Automobile Inventory:} In this example we model an inventory of an automobile company. Let $n$ denote the current number of items in the inventory and $0 \le n \le n_{\max}$ holds where $n_{\max}$ denotes the maximum capacity of the inventory. Let $t$ denote the age of the inventory in months and $0 \le t \le t_{\max}$ holds where $t_{\max}$ is the total life time of the inventory. Let us assume that an unsold car in the inventory becomes cheaper every year by a discount factor as $\beta$. Every year company decides whether to manufacture a predefined number $nc$ of new cars. We assume that the number of cars sold per month, denoted by $sold$, follows a uniform probability distribution. We also assume that the value of $sold$ can vary within a small range [$sold_{\min}, sold_{\max}$] and these two constants can be obtained from the car-sale statistics. The reward function is obtained from the cost (negative reward) of manufacturing, and the price (positive reward) of selling a car. The state-space of the model is defined as $S= \langle n, t \rangle$ where $n$ denotes the number of the cars and $t$ denotes the current month. The inventory example contains the locality property; since each state ($n,t$) has transitions to the nearby states ($\langle n-sold,t+1 \rangle$,   $\langle n,t+1\rangle$ or $\langle n+nc,t+1\rangle$).  After each fiscal year, the company computes the optimal value of the inventory keeping the discount factor in mind. The property of interest is the ``optimal discounted sum" of the car inventory.

{\bf Machine Replacement:} In this example we model the machine replacement problem. The state of machine can be in $n$ different working states from $0$ to $n-1$. The state value $0$ denotes that the machine is not working and the state value $n-1$ denotes the machine is new. The time is denoted by the variable $t$ and ranges between $0$ to a predefined maximum value $tm$. The machine can be replaced at any time and the new machine costs money (assume that the machine replacement does not add any reward). The machine can get more work done when it is in a better (state value higher) state, hence earns more money. The property of interest is the "optimal discounted value" of the machine.

{\bf Network Protocol:} Let us assume that $n$ computers follow a simpler version of Ethernet protocol to send a pre-defined number $(M-1)$ of packets to a shared channel. Let $t$ denotes the time elapsed since the start of the protocol and the condition  $0 \le t \le t_{\max}$ holds, where $t_{\max}$ is the time-out limit of the protocol. The state-space of the model can be given as a tuple $S = \langle pk_1, pk_2, \ldots, pk_n \rangle$ where $0 \le pk_i \le M$ denotes the number of the packets sent from the computer $i$. If the computer $i$ sends one more packet to the channel at time $t$, then the $i$-th component of the state changes to $\min \{ (pk_i+1), M \}$. However, two or more computers can send packets to the channel at the same time frame ({\em collision}) and both packets are lost (the state of MDP does not change). After the collision, each computer waits for a random amount of time before sending it again.
Each computer will check whether the channel is busy in time frame $t$. If the channel is busy at frame $t$, the computer does not send packets at frame $(t+1)$. Otherwise, the computer has two actions - either (1) send at frame $(t+1)$, or (2) does not send. When two computers send packets to the channel at the same time frame, there is a collision and both packets are lost. After the collision, each computer waits for a random amount of time before sending it again. The waiting time for the next packet are decided by the stations following a probability distribution. Since the packets numbers are serial in numbers, the MDP model contains the locality property. The average throughput of the shared channel is measured by the percentage of the packets sent without a collision. We are interested to compute the efficiency of the protocol by computing the average or discounted throughput property.

{\bf MTBDD-based Symbolic Implementation in PRISM :} 
We have implemented both versions (with and without MLA) of symbolic discounted algorithms within the probabilistic model checker PRISM~\cite{KNP02a}. We used the MTBDD engine of PRISM, since (a) it is generally the best performing engine for MDPs; and (b) it is the only one that can scale to the size of models we are aiming towards. The current examples with quantitative objectives cannot be handled directly by PRISM, and hence we have added a new functionality of discounted reward computation in the tool PRISM. The initial partitions are picked based on the following choice. Internally, every integer variables with range $l$ are converted into $log_2(l)$ binary variables.
 If the program have $k$ binary variables, then we pick $\frac{k}{2}$ as the initial level of abstraction. We have tried two types of partitioning procedure as proposed in ~\cite{RoyPNdA08}. 

{\bf Results :} The table above summarizes the results for all case studies (with discount factor $0.9$, $\epsilon_{abs} = 0.01$ and $\epsilon_{float} = 0.0001$). The first two columns show the name and parameters of the MDP model. The third and fourth columns give the number of states and transitions for each model respectively. The remaining columns show the performance of analyzing the MDPs, using both versions of the discounted algorithms. In both cases, we give the total time required in seconds and the peak MTBDD node count. For MLA, we also show the final number of generated regions. Our results show that MLA algorithm leads to significant space savings which is the real bottleneck in analysis of large MDPs. The number of regions increases with respect to the state-space; however the increase is linear or constant in these examples. The MTBDD node count columns provide a clear view that the symbolic iterations in the value-iteration involve the whole state space and the peak node-count is higher. MLA algorithms computes the value iteration in each region in a sequential manner, hence the size of the MTBDD graph is also smaller. There is a slowdown when MLA is applied; however time is not a bottleneck in the symbolic model-checking tools like PRISM. Most case-studies that PRISM cannot handle often fail due to excessive memory requirements, not due to time. It is also clear, from the sizes of the MDPs in the table, that the symbolic version of MLA is able to handle MDPs considerably larger than were previously feasible for the explicit implementation of \cite{deAlR07}.

\vspace{-1em}
\section{Conclusion}
\vspace{-1em}
In this paper we extend the MLA technique to solve MDPs and stochastic games with quantitative objectives. MLA is particularly well-suited to problems where there is a notion of {\em locality\/} in the state space, so that it is useful to cluster states based on values, even though their transition relations may not be similar. Many inventory, planning and control problems satisfy the locality property and would benefit from the MLA technique. In the setting of inventory, planning and control problems quantitative objectives are more appropriate than qualitative objectives. We present the MLA technique based abstraction-refinement algorithm for both stochastic games and MDPs with discounted and long run objectives. To demonstrate the applicability of our algorithms we present a symbolic implementation of our algorithms in PRISM for MDPs with discounted objectives. Our experimental results show that the MLA based technique gives significant space saving over value-iteration methods.

\clearpage
\section*{Appendix}

\begin{proof} {\bf (of Theorem~\ref{thrm-basic-property}).}
Since $f_{-}$ and $f_{+}$ are bounded and monotonic, it follows that 
there exists fixpoints $v_{-}^*$ and $v_{+}^*$ of  $f_{-}$ and $f_{+}$, 
respectively. 
Since $f$ is bounded by $f_{-}$ and $f_{+}$ it follows that the fixpoint 
$v^*$ of $f$ is bounded by $v_{-}^*$ and $v_{+}^*$, respectively, i.e.,
$v_{-}^* \leq v^* \leq v_{+}^*$.  
The desired result follows.
\hfill\qed
\end{proof}

\begin{proof} {\bf (of Lemma~\ref{lemm-1-app}).}
The properties are straightforward to verify using Definition~\ref{defi-mpre}.
\hfill\qed
\end{proof}

\begin{proof} {\bf (of Lemma~\ref{lemm-2-app}).}
It follows from Lemma~\ref{lemm-1-app} that $l_{-}$ and $l_{+}$ 
satisfies the properties of functions $f_{-}$ and $f_{+}$ of 
Theorem~\ref{thrm-basic-property}, respectively.
The results then follows from Theorem~\ref{thrm-basic-property}.
\hfill\qed
\end{proof}

\begin{proof} {\bf (of Lemma~\ref{lemm-3-app}).}
The result is easy to using Definition~\ref{defi-mpre}
and Definition~\ref{defi-mprex}.
\hfill\qed
\end{proof}

\begin{proof} {\bf (of Theorem~4 and Theorem~5).}
Both the proofs are based on the fact that as $\veabs$ and $\vefloat$ 
converges to $0$, the output of the MLA algorithm converges to the 
value of the game.
\hfill\qed
\end{proof}

\begin{proof} {\bf (Lemma~\ref{lemm-mag-rel-val}).}
It follows from definition that for all valuations 
$v$ we have 
\[
\MPre(\min,v,R) \leq \Pre(v) \leq \MPre(\max,v,R).
\]
The result follows from the above inequalities and the results of 
Lemma~\ref{lemm-rel-val}.
\hfill\qed
\end{proof}


\begin{thebibliography}{10}

\bibitem{Bertsekas95}
D.P. Bertsekas.
\newblock {\em Dynamic Programming and Optimal Control}.
\newblock Athena Scientific, 1995.
\newblock Volumes I and II.

\bibitem{CdAH-icalp05}
K.~Chatterjee, L.~de~Alfaro, and T.A. Henzinger.
\newblock The complexity of stochastic rabin and streett games.
\newblock In {\em Proc. 32nd Int. Colloq. Aut. Lang. Prog.}, volume 3580 of
  {\em LNCS}, pages 878--890. Springer, 2005.

\bibitem{CHJM05}
Krishnendu Chatterjee, Thomas~A. Henzinger, Ranjit Jhala, and Rupak Majumdar.
\newblock Counterexample-guided planning.
\newblock In {\em UAI}, July 2005.

\bibitem{CounterexCAV00}
E.~Clarke, O.~Grumberg, Y.~Lu, and H.~Veith.
\newblock Counterexample-guided abstraction refinement.
\newblock In {\em CAV 00}, LNCS. Springer, 2000.

\bibitem{Condon92}
A.~Condon.
\newblock The complexity of stochastic games.
\newblock {\em Information and Computation}, 96:203--224, 1992.

\bibitem{CY95}
C.~Courcoubetis and M.~Yannakakis.
\newblock The complexity of probabilistic verification.
\newblock {\em J. ACM}, 42(4):857--907, 1995.

\bibitem{luca-thesis}
L.~de~Alfaro.
\newblock {\em Formal Verification of Probabilistic Systems}.
\newblock PhD thesis, Stanford University, 1997.

\bibitem{deAlR07}
Luca de~Alfaro and Pritam Roy.
\newblock Magnifying-lens abstraction for {Markov} decision processes.
\newblock In {\em Proc. CAV'07}, volume 4590 of {\em LNCS}, pages 325--338.
  Springer, 2007.

\bibitem{Derman}
C.~Derman.
\newblock {\em Finite State {Markovian} Decision Processes}.
\newblock Academic Press, 1970.

\bibitem{FilarVrieze97}
J.~Filar and K.~Vrieze.
\newblock {\em Competitive {Markov} Decision Processes}.
\newblock Springer, 1997.

\bibitem{KNP02a}
M.~Kwiatkowska, G.~Norman, and D.~Parker.
\newblock {PRISM}: Probabilistic symbolic model checker.
\newblock In {\em TOOLS'02}, volume 2324 of {\em LNCS}, pages 200--204.
  Springer, 2002.

\bibitem{LigLipp69}
T.~Liggett and S.~Lippman.
\newblock Stochastic games with perfect information and time average payoff.
\newblock {\em SIAM Review}, 11:604--607, 1969.

\bibitem{Manne60}
A.~S. Manne.
\newblock Linear programming and sequential decisions.
\newblock {\em Manag Sci}, 6:259�267, 1960.

\bibitem{Rothblum78}
U.G. Rothblum.
\newblock Solving stopping stochastic games by maximizing a linear function
  subject to quadratic constraints.
\newblock {\em Game theory and related topics}, 1978.

\bibitem{RoyPNdA08}
Pritam Roy, David Parker, Gethin Norman, and Luca de~Alfaro.
\newblock Symbolic magnifying lens abstraction in markov decision processes.
\newblock In {\em QEST}, pages 103--112, 2008.

\bibitem{KNP04a}
J.~Rutten, M.~Kwiatkowska, G.~Norman, and D.~Parker.
\newblock {\em Mathematical Techniques for Analyzing Concurrent and
  Probabilistic Systems}, volume~23 of {\em CRM Monograph Series}.
\newblock AMS, 2004.

\bibitem{SegalaT}
R.~Segala.
\newblock {\em Modeling and Verification of Randomized Distributed Real-Time
  Systems}.
\newblock PhD thesis, MIT, 1995.

\end{thebibliography}

\end{document}